\documentclass{llncs}

\usepackage[utf8]{inputenc}
\usepackage{amsmath}
\usepackage{xspace}
\usepackage{tikz}
\usepackage{algorithm}
\usepackage[noend]{algpseudocode}
\algnewcommand\algorithmicforeach{\textbf{for each}}
\algdef{S}[FOR]{ForEach}[1]{\algorithmicforeach\ #1\ \algorithmicdo}
\algdef{SE}[DOWHILE]{Do}{doWhile}{\algorithmicdo}[1]{\algorithmicwhile\ #1}
\usepackage{hyperref}
\usepackage{url}
\usepackage{amssymb}
\usepackage{stmaryrd}
\usepackage{breakcites}
\usepackage{ifthen}

\usepackage[sectionbib,numbers,sort&compress]{natbib}

\usepackage{listings}
\usepackage{textcomp}

\definecolor{ciaoframe}     {rgb}{  0,    0,  0.3}
\definecolor{ciaostring}    {rgb}{0.6, 0.46, 0.33}
\definecolor{ciaooperators} {rgb}{0.1, 0.15,  0.6}
\definecolor{ciaokeywords}  {rgb}{0.1, 0.15,  0.6}
\definecolor{ciaoassertions}{rgb}{0.1, 0.15,  0.6}
\definecolor{ciaotrust}     {RGB}{200, 130,     0}
\definecolor{ciaocheck}     {rgb}{0.1, 0.2,   0.8}
\definecolor{ciaochecked}   {rgb}{0.2, 0.34,  0.1}
\definecolor{ciaotrue}      {rgb}{0.2, 0.34,  0.1}
\definecolor{ciaofalse}     {rgb}{0.6,  0.0, 0.09}
\definecolor{ciaoprops}     {rgb}{0.1,  0.2,  0.8}
\definecolor{ciaocomment}   {rgb}{0.5,  0.5,  0.5}

\newcommand{\prettylstciao}[0]{
\lstset{language=Prolog,
  frameround=fttt,
  frame=ltrb,
  rulecolor=\color{ciaoframe},
  numbers=left,numberstyle=\tiny,stepnumber=1,numbersep=8pt,
  tabsize=4,
  breaklines=true,breakatwhitespace=true,
  basicstyle=\scriptsize\ttfamily, % basic font setting
  showlines=true,
  showspaces=false,
  showtabs=false,
  escapechar=@,
  escapeinside={~~},
  commentstyle=\color{ciaocomment},
  stringstyle=\color{ciaostring},
  showstringspaces=false,
  deletekeywords={true}, % Defined in Prolog defaults
  keywordstyle={\color{ciaooperators}\bfseries}, % For the default language keywords
  classoffset=1, % starting new class
        otherkeywords={>,<,>=,=<,.,;,-,!,=,*,\&,+,:-,[,],|,->,:,:=,\#},
        keywordstyle={\color{ciaokeywords}\bfseries},
  classoffset=2,
       morekeywords={module,use_module,dynamic,export,import,impl_defined},
       keywordstyle={\color{ciaokeywords}\bfseries},
       morekeywords={pred,prop,calls,success,comp},
       keywordstyle={\color{ciaoassertions}\bfseries},
  classoffset=4,
       morekeywords={trust,trust_default,entry},
       keywordstyle={\color{ciaotrust}\bfseries},
  classoffset=5,
       morekeywords={check},
       keywordstyle={\color{ciaocheck}\bfseries},
  classoffset=6,
       morekeywords={checked},
       keywordstyle={\color{ciaochecked}\bfseries},
  classoffset=7,
       morekeywords={true},
       keywordstyle={\color{ciaotrue}\bfseries},
  classoffset=8,
       morekeywords={false},
       keywordstyle={\color{ciaofalse}\bfseries},
  classoffset=9,
       morekeywords={even,nat,int,flt,atm,term,num,var,list,ground,mshare,
                    rsize,cardinality,not_fails,exp,cost,costb,steps_ub,steps_lb,
                    size_ub,size_lb,covered,mut_exclusive,head_cost,literal_cost,
                    is_det,length,terminates,steps_o,resource,socket,seff,string
       },
       keywordstyle={\color{ciaoprops}\bfseries},
  classoffset=0, % Restore default
}}

\def\tuple#1{\langle{#1}\rangle}
    
\newcommand{\kbd}[1]{\mbox{\tt #1}}

\newcommand{\Q}{\mbox{$\cal Q$}}
\newcommand{\andtree}{generalized {\sc and} tree\xspace}
\newcommand{\andtrees}{generalized {\sc and} trees\xspace}
\newcommand*{\andnode}[1][G]{\ensuremath{\langle {#1},\theta^c,\theta^s
    \rangle}}
\newcommand{\csemantics}{\ensuremath{\llbracket P \rrbracket_{\Q}}}

\newcommand{\callingcontext}{\ensuremath{\mathit{calling\_context}(G,P,\Q)}}
\newcommand{\answers}{\ensuremath{\mathit{answers}}}

\newcommand{\DD}{\mbox{$D_\alpha$}}
\newcommand{\CD}{\mbox{$D$}}

\newcommand{\calls}{\ensuremath{\mathtt{calls}}\xspace}
\newcommand{\success}{\ensuremath{\mathtt{success}}\xspace}
\newcommand{\tsp}[1]{\ensuremath{\lambda^+_{TS(#1)}}\xspace}
\newcommand{\tsl}[1]{\ensuremath{\lambda^-_{TS(#1)}}\xspace}

\newcommand*{\seq}[2][n]  {{#2_{1}, \allowbreak \ldots, \allowbreak #2_{#1}}}

\newcommand{\chead}[1]{\ensuremath{#1.\mathtt{head}}\xspace}
\newcommand{\cbody}[1]{\ensuremath{#1.\mathtt{body}}\xspace}
\newcommand{\ppassrt}[2]{\ensuremath{\mathtt{#1}(#2)}\xspace}

\newcommand{\trustf}[1]{\emph{\color{orange}{#1}}}
\newcommand{\errors}{\ensuremath{E}\xspace}
\newcommand{\sample}{\texttt{sample-check}\xspace}
\newcommand{\conds}{\ensuremath{C}\xspace}
\newcommand{\cl}{\ensuremath{\mathsf{cl}}\xspace}
\newcommand{\PLAIs}{PLAI-simp\xspace}
\newcommand{\project}{abs\_project}
\newcommand{\extend}{abs\_extend}
\newcommand{\accel}{speed-up}
\newcommand{\pre}{\ensuremath{\mathit{Pre}}\xspace}
\newcommand{\post}{\ensuremath{\mathit{Post}}\xspace}
\newcommand{\head}{\ensuremath{\mathit{Head}}\xspace}
\newcommand{\status}{\ensuremath{\mathit{Status}}\xspace}
\newcommand{\cond}{\ensuremath{\mathit{Cond}}\xspace}
\newcommand{\call}{\ensuremath{\mathit{Call}}\xspace}
\newcommand{\CDP}[2]{\ensuremath{#1\text{\,:\,}#2}}

\usepackage{lineno}

\title{Multivariant Assertion-based Guidance in Abstract
  Interpretation%
\thanks{Research partially funded by Spanish MINECO
     grant TIN2015-67522-C3-1-R \emph{TRACES}, the Madrid M141047003
     \emph{N-GREENS} program, and Spanish MECD grant FPU16/04811. We thank the
     anonymous reviewers for their useful comments.}}

\author{Isabel Garcia-Contreras\inst{1,2} \and Jose F. Morales\inst{1} \and
  Manuel V. Hermenegildo\inst{1,2}}

\institute{IMDEA Software Institute \and
  Universidad Politécnica de Madrid (UPM)}

\begin{document}

\maketitle

\vspace{-5mm}

\begin{abstract}
  Approximations during program analysis are a necessary evil, as they ensure
  essential properties, such as soundness and termination of the analysis, but
  they also imply not always producing useful results.
  Automatic techniques have been studied to prevent precision loss, typically at
  the expense of larger resource consumption.
  In both cases (i.e., when analysis produces inaccurate results and when
  resource consumption is too high), it is necessary to have some means for
  users to provide information to guide analysis and thus improve precision
  and/or performance.
  We present techniques for supporting within an abstract interpretation
  framework a rich set of assertions that can deal with
  multivariance/context-sensitivity, and can handle different run-time semantics
  for those assertions that cannot be discharged at compile time.
  We show how the proposed approach can be applied to both improving precision
  and accelerating analysis.
  We also provide some formal results on the effects of such assertions on the
  analysis results.
\end{abstract}

\begin{keywords}
  Program Analysis \and Multivariance \and Context Sensitivity \and
  Abstract Interpretation \and Assertions \and Static Analysis \and User Guidance
\end{keywords}

\vspace{-1mm}
\section{Introduction}

Abstract Interpretation~\cite{Cousot77-short} is a well-established technique
for performing static analyses to determine properties of programs. It allows
inferring at compile-time and in finite time information that is guaranteed to
hold for all program executions corresponding to all possible sets of inputs to
the program. Reasoning about these generally infinite sets of inputs and program
paths requires \emph{(safe) approximations} --computing over \emph{abstract
  domains}-- to ensure termination and soundness. If such approximations are not
carefully designed, the information reported by the analyzer may not be accurate
enough to be useful for the intended application, such as, e.g., performing
optimizations or verifying properties. Similarly, although abstract
interpretation-based analyzers are guaranteed to terminate, this does not
necessarily imply that they do so in acceptable time or space, i.e., their
resource usage may be higher than desirable.

Much work has been done towards improving both the accuracy and efficiency of
analyzers through the design of automatic analysis techniques that include
clever abstract domains, widening and narrowing
techniques~\cite{BagnaraHZ04-shortest,
  eterms-sas02-short,Zaffanella99widening-short}, and sophisticated fixpoint
algorithms~\cite{ bourd90-short, abs-int-naclp89-short, inc-fixp-sas-short,
  clpr-anal-short }.
Despite these advances, there are still cases where it is necessary for the user
to provide input to the analyzer to guide the process in order to regain
accuracy, prevent imprecision from propagating, and improve analyzer performance
~\cite{full-prolog-esop96-short, DBLP:conf/sas/DelmasS07-short}.
Interestingly, there is comparatively little information on these aspects of
analyzers, perhaps because they are perceived as internal or analyzer
implementation-specific.

In this paper we focus on techniques that provide a means for the programmer to
be able to optionally annotate program parts in which precision needs to be
recovered. Examples are the \emph{entry} and \emph{trust} declarations of
CiaoPP~\cite{full-prolog-esop96-short,assert-lang-disciplbook-short} and the
\emph{known facts} of
Astr\'{e}e~\cite{astree-esop05-short,DBLP:conf/sas/DelmasS07-short} (see
Sect.~\ref{sec:related-work} for more related work).
Such user annotations allow dealing with program constructs for which the
analysis is not complete or the source is only partially available.
However, as mentioned before, there is little information in the
literature on these assertions beyond a sentence or two in the user
manuals or some examples of use in demo sessions.  In particular, no
precise descriptions exist on how these assertions affect the
analysis process and its results.

We clarify these points by proposing a user-guided multivariant fixpoint
algorithm that makes use of information contained in different kinds of
assertions, and provide formal results on the influence of such assertions on
the analysis. We also extend the semantics of the assertions to control if
precision can be relaxed, and also to deal with both the cases in which the
program execution will and will not incorporate run-time tests for unverified
assertions. Note that almost all current abstract interpretation systems assume
in their semantics that the run-time checks will be run. However, due to
efficiency considerations, assertion checking in often turned off in production
code, specially for complex properties \cite{rtchecks-cost-2018-ppdp-short}.
To the best of our knowledge this is the first precise description of
how such annotations are processed within a standard parametric and
multivariant fixpoint algorithm, and of their effects on analysis
results.

\vspace{-4mm}
\section{Preliminaries}
\vspace{-2mm}

\paragraph{\textbf{Program Analysis with Abstract Interpretation.}}
Our approach is based on {\em abstract interpretation} \cite{Cousot77-short}, a
technique in which execution of the program is simulated on an {\em abstract
  domain} ($D_\alpha$) which is simpler than the actual, {\em concrete domain}
($D$). Although not strictly required, we assume that \DD~has a lattice
structure with meet $(\sqcap)$, join $(\sqcup)$, and less than $(\sqsubseteq)$
operators. Abstract values and sets of concrete values are related via a pair of
monotonic mappings $\langle \alpha, \gamma \rangle$: {\em abstraction} $\alpha:
D\rightarrow D_\alpha$, and {\em concretization} $\gamma: D_\alpha\rightarrow
D$, which form a Galois connection. A description (or abstract value) $d \in
\DD$ \emph{approximates} a concrete value $c \in \CD$ if $\alpha(c) \sqsubseteq
d$ where $\sqsubseteq$ is the partial ordering on \DD. Concrete operations on
$D$ values are (over-)approximated by corresponding abstract operations on
\DD~values. The key result for abstract interpretation is that it guarantees
that the analysis terminates, provided that $D_\alpha$ meets some conditions
(such as finite ascending chains) and that the results are safe approximations
of the concrete semantics (provided \DD~safely approximates the concrete values
and operations).

\vspace{-3mm}
\paragraph{\textbf{Intermediate Representation.}}
For generality, we formulate our analysis to work on a block-level intermediate
representation of the program, encoded using Constrained Horn clauses (CHC). A
{\em definite CHC program}, or {\em program}, is a finite sequence of clauses.
A {\em clause} is of the form $H \mbox{\tt :-} B_1,\dots ,B_n$ where $H$, the
{\em head}, is an atom, and $B_1,\dots ,B_n$ is the {\em body}, a possibly empty
finite conjunction of atoms. Atoms are also called {\em literals}.
We will refer to the head and the body of a clause \cl with \chead{\cl} and
\cbody{\cl} respectively.
An {\em atom} is of the form $p(\seq{V})$. It is \emph{normalized} if the
$\seq{V}$ are all distinct variables. Normalized atoms are also called {\em
  predicate descriptor}s.
Each maximal set of clauses in the program with the same descriptor as head
(modulo variable renaming) 
defines a {\em predicate} (or {\em procedure}).
Body literals can be predicate descriptors, which represent 
{\em calls} to the corresponding predicates, or {\em constraints}.
A {\em constraint} is a finite conjunction of built-in relations for some
background theory. We assume that all non-builtin atoms are normalized. This is
not restrictive since programs can always be put in this form, and it simplifies
the presentation of the algorithm. However, in the examples we use
non-normalized programs.
The encoding of program semantics in CHC depends on the source language and is
beyond the scope of the paper.
It is trivial for (C)LP programs, and also well studied
for several types of imperative programs and compilation levels (e.g., bytecode,
llvm-IR, or ISA --
see~\cite{mod-decomp-scam08-shortest,HGScam06-short,decomp-oo-prolog-lopstr07-short,
  DBLP:conf/birthday/BjornerGMR15-short, DBLP:conf/cav/GurfinkelKKN15-shorter,
  isa-energy-lopstr13-final-shortest, resources-bytecode09-short}).

\vspace{-3mm}
\paragraph{\textbf{Concrete Semantics.}}
The concrete semantics that we abstract is that of Constraint Logic Programs --
(C)LP~\cite{intro_constraints_stuckey}.
In particular, we use the constraint extension of top-down,
left-to-right SLD-resolution, which, given a {\em query} ({\em initial
  state}), returns the answers ({\em exit states}) computed for it by
the program.
A \emph{query} is a pair $\CDP{G}{\theta}$ with $G$ a (non-empty) conjunction of
atoms and $\theta$ a constraint. Executing (answering) a query with respect to a
CHC program consists on determining whether the query is a logical consequence
of the program and for which constraints (answers).
However, since we are interested in abstracting the calls and answers (states)
that occur at different points in the program, we base our semantics on the
well-known notion of \andtrees~\cite{bruy91}.
The concrete semantics of a program $P$ for a given set of queries
$\Q$, \csemantics, is then the set of \andtrees\ that result from the
execution of the queries in $\Q$ for $P$.
Each node \andnode\ in the \andtree\ represents a call to a predicate $G$ (an
atom), with the constraint (state) for that call, $\theta^c$, and the
corresponding success constraint $\theta^s$ (answer).
The \callingcontext\ of a predicate given by the predicate descriptor $G$
defined in $P$ for a set of queries $\Q$ is the set $\{ \theta^c\ |\ \exists T
\in \csemantics\ s.t.\ \exists \tuple{G',\theta'^c,\theta'^s}\ in\ T \wedge
\exists\sigma, \sigma(G') = G, \sigma(\theta'^c)=\theta^c\}$, where $\sigma$ is
a \emph{renaming} substitution, i.e., a substitution that replaces each variable
in the term it is applied to with distinct, fresh variables. We use $\sigma(X)$
to denote the application of $\sigma$ to $X$.
We denote by $\answers(P,\Q)$ the set of success constraints computed by $P$ for
queries $\Q$.

\vspace{-3mm}
\paragraph{\textbf{Goal-dependent abstract interpretation.}}
We use goal-dependent abstract interpretation, in particular a simplified
version (\PLAIs) of the PLAI
algorithm~\cite{abs-int-naclp89-short,ai-jlp-short}, which is essentially an
efficient abstraction of the \andtrees~semantics, parametric on the abstract
domain. It takes as input a program $P$, an abstract domain $D_\alpha$, and a
set of abstract initial queries $\Q_{\alpha} = \{G_i\mbox{:}\lambda_i\}$, where
$G_i$ is a normalized atom, and~$\lambda_i \in D_\alpha$ is abstract constraint.
The algorithm computes a set of triples $A =$
$\{\tuple{G_1,\lambda_{1}^c,\lambda_{1}^s},$ $\ldots,$
$\tuple{G_n,\lambda_{n}^c,\lambda_{n}^s}\}$. In each
$\tuple{G_i,\lambda_{i}^c,\lambda_{i}^s}$ triple, $G_i$ is a normalized atom,
and $\lambda_{i}^c$ and $\lambda_{i}^s$, elements of $D_\alpha$, are,
respectively, abstract call and success constraints. The set of triples for a
predicate cover
all the concrete call and success constraints that appear during execution of
the initial queries from $\gamma(Q_\alpha)$, see~Def.~\ref{def:correct}.

As usual, $\bot$ denotes the abstract constraint such that
$\gamma(\bot)=\emptyset$. A tuple $\tuple{G_j,\lambda_j^c,\bot}$ indicates that
all calls to predicate $G_j$ with any constraint $\theta\in\gamma(\lambda_j^c)$
either fail or loop, i.e., they do not produce any success constraints.
$A$ represents the (possibly infinite) set of nodes of the \andtrees for the
queries represented in $\Q_\alpha$ to $P$. In addition, $A$ is multivariant on
calls, namely, it may contain more than one triple for the same predicate
descriptor $G$ with different abstract call constraints.
The PLAI algorithm provides guarantees on termination and correctness (see
Thm.~\ref{th:basic} for a more precise formulation).

\vspace{-3mm}
\paragraph{\textbf{Assertions.}}
\label{sec:assertions}

Assertions allow stating conditions on the state (current constraint store) that
hold or must hold at certain points of program execution.
We use for concreteness a subset of the syntax of the \texttt{pred} assertions
of~\cite{full-prolog-esop96-short,prog-glob-an-shorter,assert-lang-disciplbook-short},
which allow describing sets of \emph{preconditions} and \emph{conditional
  postconditions} on the state for a given predicate.
These assertions are instrumental for many purposes, e.g. 
expressing the results of analysis, providing 
specifications,
and documenting~\cite{%
  ciaopp-sas03-journal-scp-shortest,%
  prog-glob-an-shorter,%
  assrt-theoret-framework-lopstr99-short%
}.
A \texttt{pred} assertion is of the form:

\centerline{$ \kbd{:- [}\status \kbd{] pred } \head \kbd{ [: } \pre \kbd{] [=> } \post \kbd{].} $}

\noindent
where \head is a predicate descriptor (i.e., a normalized atom) that denotes the
predicate that the assertion applies to, and \pre and \post are conjunctions of
\emph{property literals}, i.e., literals corresponding to predicates meeting
certain conditions which make them amenable to checking, such as being decidable
for any input~~\cite{assert-lang-disciplbook-short}.
\pre expresses properties that hold when \head is called, namely, at least one
\pre must hold for each call to \head. \post states properties that hold if
\head is called in a state compatible with \pre and the call succeeds. Both \pre
and \post can be empty conjunctions (meaning true), and in that case they can be
omitted.
$\status$ is a qualifier of the meaning of the assertion. Here we consider: {\tt
  trust}, the assertion represents an actual behavior of the predicate that the
analyzer will assume to be correct; {\tt check}, the assertion expresses
properties that must hold at run-time, i.e., that the analyzer should prove or
else generate run-time checks for (we will return to this in
Sect.~\ref{sec:adding}). {\tt check} is the default status of assertions.

\begin{example}
  \label{ex:pow}
  The following assertions describe different behaviors of the \kbd{pow}
  predicate that computes $\mathtt{P = X^N}$: {\tt (1)} is stating that if
  the exponent of a power is an even
  number, the result ({\tt P}) is non-negative,
  {\tt (2)} states that if the base is a non-negative number and the exponent is a
  natural number the result {\tt P} also is non-negative:
\clearpage
\prettylstciao
\begin{lstlisting}
:- pred pow(X,N,P) : (int(X), even(N)) => P ~$\geq$~ 0.  % (1)
:- pred pow(X,N,P) : (X ~$\geq$~ 0, nat(N))   => P ~$\geq$~ 0.  % (2)
pow(_, 0, 1).
pow(X, N, P) :- N > 0,
    N1 is N - 1, pow(X, N1, P0), P is X * P0.
\end{lstlisting}
Here, for simplicity we assume that the properties \texttt{even/1},
\texttt{int/1}, \texttt{nat/1}, and \texttt{$\geq$} are \emph{built-in
  properties} handled natively by the abstract domain.
\end{example}

\noindent
In addition to predicate assertions we also consider \emph{program-point
  assertions}. They can appear in the places in a program in which a literal
(statement) can be added and are expressed using literals corresponding to their
$\status$, i.e., \ppassrt{trust}{\cond} and \ppassrt{check}{\cond}. They imply
that whenever the execution reaches a state originated at the program point in
which the assertion appears, $\cond$ (should) hold. Example~\ref{ex:pp_asrt}
illustrates their use.
Program-point assertions can be translated to {\tt pred} assertions,%
\footnote{%
  E.g., we can replace line \ref{line:ex_pp_trust} in Example \ref{ex:pp_asrt},
  by ``{\tt assrt\_aux(Z)},'', and add a predicate to the program, {\tt
    assrt\_aux(\_).}, with an assertion ``{\tt :- pred assrt\_aux(Z) : Z = 2.}''.}
so without loss of generality
we will limit the discussion to {\tt pred} assertions.

\begin{definition}[Meaning of a Set of Assertions for a Predicate]
  \label{def:assr_cond}
  Given a predicate represented by a normalized atom  \head, and a
corresponding set of assertions
$\{a_1 \ldots a_n\}$, with $a_i =
``\texttt{:- pred } \head \texttt{ : } \pre_i \texttt{ => } \post_i
\texttt{.}$'' the set of \emph{assertion conditions} for \head
is $\{ C_0, C_1, \ldots , C_n\}$, with:
\vspace{-1mm}
  \[
    C_i = \left\{
    \begin{array}{ll}
      \calls(\head,\bigvee _{j = 1}^{n} \pre_j)
    & ~~~~i = 0 
    \\
      \success(\head,\pre_i,\post_i)
    & ~~~~i = 1..n
    \end{array}
    \right.
  \]
\end{definition}
\noindent
where $\calls(\head,\pre)$%
\footnote{We denote the calling conditions with \calls (plural) for historic
  reasons, and to avoid confusion with the higher order predicate in Prolog {\tt call/2}.}
states conditions on all concrete calls to the
predicate described by \head,
and
$\success(\head,\pre_j,\post_j)$ describes conditions on the success constraints
produced by calls to \head if $\pre_j$ is satisfied.

The assertion conditions for the assertions in Example~\ref{ex:pow}
are:
\[ 
   \left\{
     \begin{array}{lllll}
       \calls ( & pow(X,N,P), & ((int(X), even(N)) \vee (X\geq 0, nat(N)))), & \\
       \success (& pow(X,N,P), &  (int(X), even(N)), & (P\geq 0)), \\
       \success (& pow(X,N,P), & (X\geq 0, nat(N)), & (P\geq 0)) \\
    \end{array}
    \right\}
\]

\paragraph{\textbf{Uses of assertions.}}

We show 
examples of the use assertions to guide analysis.

\begin{example}\label{ex:pp_asrt}
  \emph{Regaining precision during analysis.}
  \noindent
  If we analyze the following program with a simple (non-relational) intervals
  domain, the information inferred for {\tt Z} would be ``any integer''
  (line~\ref{line:ex_loss}), whereas
  it can be seen
  that it is $Z=2$ for any {\tt X} and
  {\tt Y}.
  We provide the information to the analyzer with an assertion
  (line~\ref{line:ex_pp_trust}). The analyzer will \emph{trust} this information
  even if it cannot be inferred with this domain (because it cannot represent
  relations between variables). \prettylstciao
\begin{lstlisting}
p(Y) :-           % (Y > 0)
    X is Y + 2,   % (X > 2,  Y > 0)
    Z is X - Y,   % (int(Z), X > 2, Y > 0) ~\label{line:ex_loss}~
    trust(Z = 2), % (Z = 2,  X > 2, Y > 0) ~\label{line:ex_pp_trust}~
    % implementation continues
\end{lstlisting}
\end{example}

\begin{example}\emph{Speeding up analysis.}
  \noindent
  Very precise domains suffer less from loss of precision and are useful for
  proving complex properties, but can be very costly. In some cases less precise
  information in enough, e.g., this code extracted from LPdoc, the Ciao
  documentation generator, {\tt html\_escape} is a predicate that takes a string
  of characters and transforms it to html: \prettylstciao
  \begin{lstlisting}
:- trust pred html_escape(S0, S) => (string(S0), string(S)).
html_escape("``"||S0, "&ldquo;"||S) :- !, html_escape(S0, S).
html_escape("''"||S0, "&rdquo;"||S) :- !, html_escape(S0, S).
html_escape([34|S0],  "&quot;"||S)  :- !, html_escape(S0, S).
html_escape([39|S0],  "&apos;"||S)  :- !, html_escape(S0, S).
% ...
html_escape([X|S0], [X|S])          :- !, character_code(X), html_escape(S0, S).
html_escape([],[]).

% string(Str) :- list(Str, int).
\end{lstlisting}
  Analyses based on regular term languages, as, e.g.
  \kbd{eterms}~\cite{eterms-sas02-short} infer precise regular types with
  subtyping, which is often costly.
  In this example it would be equivalent to computing an accurate regular
  language that over-approximates the HTML text encoding.
  The {\tt trust} assertion provides a general invariant that the analyzer will
  take instead of inferring a more complex type.
\end{example}
\begin{example}\emph{Defining abstract usage or specifications of libraries or
    dynamic predicates}.
  When sources are not available, or cannot be analyzed, assertions can provide
  the missing abstract semantics. The following code illustrate the use of an
  assertion to describe the behavior of predicate {\tt receive} in a {\tt
    sockets} library that is written in C. The assertion in this case
  transcribes what is stated in natural language in the documentation of the
  library. Note that if no annotations were made, the analyzer would have to
  assume the most general abstraction ($\top$) for the library arguments.

  \prettylstciao
  \begin{lstlisting}
:- module(sockets, []).

:- export(receive/2).
:- pred receive(S, M) : (socket(S), var(M)) => list(M, utf8).
% receive is written in C
\end{lstlisting}
\end{example}
\begin{example}\emph{(Re)defining the language semantics for abstract domains.\label{ex:product}}
   \noindent
   \kbd{trust} assertions are also a useful tool for defining the
   meaning (transfer function)
   of the basic operations of the language. In this
   example we define some basic properties of the product predicate
   in a simple types-style abstract domain:
  \prettylstciao
  \begin{lstlisting}
:- trust pred '*'(A, B, C) : (int(A), int(B)) => int(C).
:- trust pred '*'(A, B, C) : (flt(A), int(B)) => flt(C).
:- trust pred '*'(A, B, C) : (int(A), flt(B)) => flt(C).
:- trust pred '*'(A, B, C) : (flt(A), flt(B)) => flt(C).
  \end{lstlisting} 
  The semantics of bytecodes or machine instructions can be specified for each
  domain after transformation into CHCs. Assertions allow representing behaviors
  for the same predicate for different call descriptions (multivariance).
\end{example}

\vspace{-4mm}
\section{Basic fixpoint algorithm}
\vspace{-2mm}

We first present a basic, non-guided algorithm to be used as starting point
--see Fig.~\ref{alg:naive}. \PLAIs is essentially the PLAI
algorithm~\cite{ai-jlp-short}, but omitting some optimizations that are
independent from the issues related with the guidance. The algorithm is
parametric on the abstract domain \DD, given by implementing the
domain-dependent operations $\sqsubseteq, \sqcap, \sqcup,$ {\tt abs\_call}, {\tt
  abs\_proceed}, {\tt abs\_generalize}, {\tt \project}, and {\tt \extend} (which
will be described later), and transfer functions for program built-ins, that
abstract the meaning of the basic operations of the language. These operations
are assumed to be monotonic and to correctly over-approximate their
correspondent concrete version.
As stated before, the goal of the analyzer is to capture the behavior of each
procedure (function or predicate) in the program with a set $A$ of triples
$\tuple{G,\lambda^c,\lambda^s}$, where $G$ is a normalized atom and $\lambda^c$
and $\lambda^s$ are, respectively, the abstract call and success constraints,
elements of $D_\alpha$.
For conciseness, we denote looking up in $A$ with a partial function $a: Atom *
\DD \mapsto \DD$, where $\lambda^s = a[G,\lambda^c] ~\mathit{iff}~ \tuple{G,
  \lambda^c, \lambda^s} \in A$, and modify the value of $a$ for $(G,\lambda^c)$,
denoted with $a[G,\lambda^c] \leftarrow \lambda^{s'}$ by removing $\tuple{G,
  \lambda^c, \_}$ from $A$ and inserting $\tuple{G, \lambda^c, \lambda^{s'}}$.
In $A$ there may be more than one triple with the same $G$, capturing
multivariance, but only one for each $\lambda^c$ during the algorithm's
execution or in the final results.
  
\paragraph{\textbf{Operation of the algorithm.}}

\newcommand{\changes}{\ensuremath{changes}\xspace}
\begin{figure}[t]
    \textsc{Algorithm} \textbf{\textsc{Analyze}}$(P, \Q_\alpha)$
    
    \textbf{input:} $P, \Q_\alpha$
    \textbf{output global:} $A \leftarrow \emptyset$

    \begin{algorithmic}[1]
      \State $a[L_i,\lambda_i] \leftarrow \bot$ {\bf for all} $\CDP{L_i}{\lambda_i} \in
      \Q_\alpha $, \changes $\leftarrow$ true
      \Comment{Initial queries}
      
      \While{\changes}

      \State \changes $\leftarrow$ false
      \State $W \leftarrow \{(G,\lambda^c,\cl)\ |\
      a[G,\lambda^c]$ is defined $\wedge\ \cl \in P\ \wedge \exists
      \sigma$ s.t. $G = \sigma(\cl.{\tt head})\}$
      
      \ForEach{$(G,\lambda^c,\cl) \in W$}
      \State $\lambda^t \leftarrow$ {\tt abs\_call}$(G,\lambda^c,\cl.{\tt head})$
      \State $\lambda^t \leftarrow$ {\tt solve\_body}$(\cl.{\tt body}, \lambda^t)$
      \State {$\lambda^{s_0} \leftarrow$ {\tt abs\_proceed}$(G,\cl.{\tt
          head},\lambda^t)$}
      \State $\lambda^{s'} \leftarrow {\tt
        abs\_generalize}(\lambda^{s_0}, \{a[G,\lambda^c]\})$

      \If{$\lambda^{s'} \neq \lambda^s$}
      \State $a[G,\lambda^c] \leftarrow \lambda^{s'}$, $\changes \leftarrow$ true
      \Comment{Fixpoint not reached yet}
      \EndIf \EndFor
      \EndWhile
      
      \Function{\tt solve\_body}{$B, \lambda^t$}
      \ForEach {$L \in B$}
      \State {$\lambda^c \leftarrow$ {\tt \project}$(L,\lambda^t)$}
      \State $\call = \{\lambda\ |\ a[H,\lambda'] $ is defined $\wedge\ \exists
      \sigma$ s.t. $\sigma(H)=L \wedge \lambda = \sigma(\lambda') \}$
      
      \State $\lambda^{c'} \leftarrow$ {\tt abs\_generalize}$(\lambda^c, \call)$
    \State $\lambda^s \leftarrow $ {\tt solve}$(L,\lambda^{c'})$
    \State $\lambda^t \leftarrow$ {\tt \extend}$(L,\lambda^s,\lambda^t)$
      \EndFor
      \State \Return $\lambda^t$

      \EndFunction
    \end{algorithmic}
\vspace*{-2mm}
        \caption{Baseline fixpoint analysis algorithm (\PLAIs).}\label{alg:naive}
\vspace*{-7mm}
\end{figure}
    
Analysis proceeds from the initial abstract queries
$\Q_\alpha$ assuming $\bot$ as under-approximation of their success constraint. 
The algorithm iterates over possibly incomplete results (in $A$), recomputing
them with any newly inferred information, until a global fixpoint is reached
(controlled by flag $\changes$).
First, the set of captured call patterns and the clauses whose head applies
(i.e., there exists a renaming $\sigma$ s.t.\ $G = \sigma(\cl.{\tt head})$) is
stored in $W$.
Then, each clause is solved with the following process. An ``abstract
unification'' ({\tt abs\_call}) is made, which performs the abstract parameter
passing. It includes \emph{renaming} the variables, abstracting the parameter
values (via function $\alpha$), and \emph{extending} the abstract constraint to
all variables present in the head and the body of the clause.
To abstractly execute a clause the function {\tt solve\_body} abstractly
executes each of the literals of the body. This implies, for each literal, {\em
  projecting} the abstract constraint onto the variables of the literal ({\tt
  \project}) and {\em generalizing} it if necessary ({\tt abs\_generalize})
before calling {\tt solve}. Generalization is necessary to ensure termination
since we support multivariance and infinite domains. Lastly, after returning
from {\tt solve} (returning from the literal call), {\tt \extend} propagates the
information given by $\lambda^s$ (success abstract constraint over the variables
of $L$) to the constraint of the variables of the clause $\lambda^t$.
The {\tt solve} function executes abstractly a literal (Fig.~\ref{fig:solve}).
Depending on the nature of the literal, different actions will be
performed.
For \emph{built-in operations}, the corresponding transfer function ($f^\alpha$)
is applied. For \emph{predicates defined in the program}, the answer is first
looked up in $A$.
If there is already a computed tuple that matches the abstract call, the
previously inferred result is taken. Else (no stored tuple matches the abstract
call), an entry with that call pattern and $\bot$ as success value is
added. 
  This will
  trigger the analysis of this call in the next iteration of the loop.
\begin{figure}[t]
  \textbf{Global:} $A$
  \begin{algorithmic}[1]
    \Function {\tt solve}{$L, \lambda$}

    \If{$L$ is a {\em built-in}}
    \State \Return $f^\alpha(L,\lambda)$ \Comment{apply transfer function}
    \ElsIf{$a[G,\lambda^c]$ is defined and $\exists \sigma$ s.t. $\sigma(G) = L$}
    \State \Return {$\sigma(a[G,\lambda^c])$}
    \Else 
    \State $a[L, \lambda] \leftarrow \bot$ \label{alg:ins_call}
    \State \Return $\bot$
    \EndIf
    \EndFunction
  \end{algorithmic}
\vspace*{-2mm}
  \caption{Pseudocode for solving a literal.}
  \label{fig:solve}
\vspace{-4mm}
\end{figure}
Once a body is processed, the actions of {\tt abs\_call} have to be
undone in {\tt abs\_proceed}, which performs the ``abstract return''
from the clause. It \emph{projects} the temporary abstract
constraint (used to solve the body) back to the variables in the
head of the clause and \emph{renames} the resulting abstract
constraint back to the variables of the analyzed head.
The result is then abstractly \textit{generalized} with
the previous results (either from other clauses that also unify or
from previous results of the processed clause), and it is compared
with the previous result to check whether the fixpoint was
reached. Termination is ensured even in the case of domains with
infinite ascending chains because {\tt abs\_generalize} includes
performing a widening if needed, in addition to the join operation $\sqcup$.
This process is repeated for all the tuples of the analysis until the
analysis results are the same in two consecutive iterations.

Fig.~\ref{fig:fact} shows a factorial program and an analysis result $A$ for
$\Q_\alpha = \{\mathtt{fact}(X,R):\top\}$ with an abstract domain that keeps
information about signs for each of the program variables with values of the
lattice shown. For example, the first tuple in $A$ states that {\tt fact(X,R)}
may be called with any possible input and, if it succeeds, $X$ will be an
integer and $R$ will be a positive number.

\begin{figure}[t]
\vspace*{-2mm}
\centering
\hspace{-3mm}  \begin{minipage}{0.21\linewidth}
    \prettylstciao
\begin{lstlisting}
~\label{line:base}~fact(0,1).
~\label{line:rec}~fact(N,R) :- 
    N > 0,
    N1 is N - 1,
    fact(N1, R1),
    R is N * R1.
\end{lstlisting}  
  \end{minipage}
    \hspace{5mm}
  \begin{minipage}{0.4\linewidth}
    \begin{tabular}{llll}
      & $A = $ & &
      \\ $\{$ & $\langle fact(X,R),$ & $(X/\top, R/\top),$ & $(X/int, R/+)\rangle$
      \\ & $\langle fact(X,R),$ & $(X/int, R/\top),$ & $(X/int, R/+)  \rangle\}$ \\
\end{tabular}
  \end{minipage}
\hspace{13mm}
\begin{minipage}{0.2\linewidth}
  \begin{tikzpicture}[node distance=0.8cm]
    \node(top)                           {$\top$};
    \node(int)       [below of=top] {$int$};
    \node(z)      [below of=int]  {$0$};
    \node(l)      [below left of=int]       {$-$};
    \node(p)      [below right of=int]       {$+$};
    \node(bot)            [below of=z]     {$\bot$};
    \draw(top)       -- (int);
    \draw(int)       -- (z);
    \draw(int)       -- (p);
    \draw(int)       -- (l);
    \draw(l)      --  (bot);
    \draw(p)      --  (bot);
    \draw(z)      --  (bot);
  \end{tikzpicture}
\end{minipage}
\vspace{-4mm}
\caption{Factorial program and a possible analysis result.\label{fig:fact}}
\vspace{-7mm}
\end{figure}
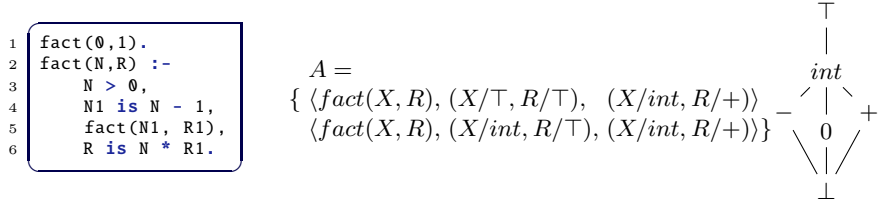

We define analysis results to be correct if the abstract call
constraints cover all the call constraints (and, respectively, the
abstract success constraints cover all the success constraints) which
appear during the concrete execution of the initial queries in
$\Q$. Formally:
\begin{definition}[Correct analysis]\label{def:correct}
  Given a program $P$ and initial queries $\Q$, an analysis result $A$ is
  correct for $P, \Q$ if:
  \begin{itemize}
  \item $\forall G, \theta^c \in \callingcontext$
    $\exists \tuple{G,\lambda^c,\lambda^s}\in A$
    s.t.\ $\theta^c \in \gamma(\lambda^c)$.
  \item $\forall \tuple{G,\lambda^c,\lambda^s} \in A, \forall
    \theta^c\in\gamma(\lambda^c)$ if $\theta^s \in \answers(P,
    \{\CDP{G}{\theta^c}\})$
    then $\theta^s\in\gamma(\lambda^s)$.
  \end{itemize}  
\end{definition}

\noindent
We recall the result from~\cite{ai-jlp-short}, adapted to the notation used in this paper.
\begin{theorem}{\bf Correctness of PLAI.}\label{th:basic}
  Consider a program $P$ and a set of initial abstract queries $\Q_\alpha$.
  Let $\Q$ be the set of concrete queries: $\Q = \{\CDP{G}{\theta} \mid\
  \theta\in\gamma(\lambda) \wedge \CDP{G}{\lambda} \in \Q_\alpha\}$.
  \noindent
  The analysis result $A = \{\tuple{G_1,\lambda_1^c, \lambda_1^s}, \ldots,
  \tuple{G_n,\lambda_n^c,\lambda_n^s}\}$ for $P$ with
  $\Q_\alpha$ is \emph{correct} for $P, \Q$.
\end{theorem}

\vspace{-4mm}
 \section{Adding assertion-based guidance to the algorithm}
\vspace{-2mm}
\label{sec:adding}

We now address how to apply the guidance provided by the user in the analysis
algorithm. But before that we make some observations related to the run-time
behavior of assertions.

\paragraph{\textbf{Run-time semantics of assertions.}}\label{sec:runtime-sem}
Most systems make assumptions during analysis with respect to the run-time
semantics of assertions: for example, Astr\'{e}e assumes that they are always
run, while CiaoPP assumes conservatively that they may not be (because in
general they may in fact be disabled by the user). In order to offer the user
the flexibility of expressing these different situations we introduce a new
status for assertions, \sample, as well as a corresponding program-point
assertion, \kbd{\sample\hspace*{-2mm}($\cond$)}.
This \sample status indicates that the properties in these assertions
may or may not be checked during execution, i.e., run-time checking
can be turned on or off (or done intermittently)
for them.
In contrast, for \kbd{check} assertions (provided that they have not been
discharged statically) run-time checks must always be performed.
\begin{table}[t]
  \centering
  \begin{tabular}{|r|c|c|}
    \hline
    \status & Use in analyzer & Run-time test  \\
             & & (if not discharged at compile-time) \\
    \hline
    {\tt trust} & yes & no  \\
    {\tt check} & yes & yes \\
    \sample     & no  & optional \\
    \hline
  \end{tabular}
  \caption{Usage of assertions during analysis.}\label{tab:runtime}
  \vspace*{-7mm}
\end{table}
Table \ref{tab:runtime} summarizes this behavior with respect to whether
run-time testing will be performed and whether the analysis can ``trust'' the
information in the assertion, depending on its status.
The information in {\tt trust} assertions is used by the analyzer but
they are never checked at run time. {\tt check} assertions are also
checked at run time and the execution will not pass beyond that point
if the conditions are not met.\footnote{This strict run-time semantics
  for {\tt check} assertions was used in \cite{optchk-journal-scp-short}.}
This means that {\tt check} assertions can also be ``trusted,'' in a
similar way to {\tt trust} assertions, because execution only
proceeds beyond them if they hold.
Finally, \sample assertions may or may not be checked at run-time (e.g., for
efficiency reasons) and thus they cannot be used as \kbd{trust}s during
analysis.

\paragraph{\textbf{Correctly applying guidance.}}

We recall some definitions (adapted from \cite{
  assrt-theoret-framework-lopstr99-short}) which are instrumental to correctly
approximate the properties of the assertions during the guidance.

\begin{definition}[Set of Calls for which a Property Formula Trivially
  Succeeds (Trivial Success Set)]
  \label{def:trivial-suc-set}
  Given a conjunction $L$ of property literals and the definitions for each of
  these properties in $P$, we define the {\em trivial success set} of $L$ in
  $P$ as:\\
  \centerline{ $ TS(L,P)= \{ \theta|Var(L) \ s.t.\ \exists \theta'\in
    \answers(P,\{\CDP{L}{\theta}\}), \theta\models\theta' \} $ }
\end{definition} 
where $\theta|Var(L)$ above denotes the projection of $\theta$ onto the
variables of $L$, and $\models$ denotes that $\theta'$ is a more general
constraint than $\theta$ (entailment). Intuitively, $TS(L,P)$ is the set of
constraints $\theta$ for which the literal $L$ succeeds without adding new
constraints to $\theta$ (i.e., without constraining it further).
For example, given the following program $P$:
  \prettylstciao
\begin{lstlisting}
list([]).
list([_|T]) :- list(T).
\end{lstlisting}
  and $L=list(X)$, both $\theta_1 = \{ X = [1,2] \}$ and $\theta_2 = \{ X =
  [1,A] \}$ are in the trivial success set of $L$ in $P$, since calling
  $(X=[1,2],list(X))$ returns $X=[1,2]$ and calling $(X=[1,A],list(X))$ returns
  $X=[1,A]$. However, $\theta_3 = \{ X = [1|\_] \}$ is not, since a call to $(X
  = [1|Y], list(X))$ will further constrain the term $[1|Y]$, returning $X =
  [1|Y], Y=[]$.
  We define abstract counterparts for Def.~\ref{def:trivial-suc-set}:

\begin{definition}[Abstract Trivial Success Subset of a Property Formula]
  \label{def:abstract-trivial-suc-set}
  Under the same conditions of Def. \ref{def:trivial-suc-set}, given an abstract
  domain $\DD$, $\lambda^-_{TS(L,P)} \in \DD$ is an {\em abstract trivial
    success subset} of $L$ in $P$ iff $\gamma(\lambda^-_{TS(L,P)})\subseteq
  TS(L,P)$.
\end{definition}

\begin{definition}[Abstract Trivial Success Superset of a Property Formula]
  Under the same conditions of Def.~\ref{def:abstract-trivial-suc-set}, an
  abstract constraint $\lambda^+_{TS(L,P)}$ is an {\em abstract trivial success
    superset} of $L$ in $P$ iff $\gamma(\lambda^+_{TS(L,P)})\supseteq TS(L,P)$.
\end{definition}
\noindent
I.e., $\lambda^-_{TS(L,P)}$ and $\lambda^+_{TS(L,P)}$ are, respectively,
safe under- and over-approximations of $TS(L,P)$.
These abstractions come useful when the properties expressed in the assertions
cannot be represented exactly in the abstract domain. Note that they are always
computable by choosing the closest element in the abstract domain, and at the
limit $\bot$ is a trivial success subset of any property formula and $\top$ is a
trivial success superset of any property formula.

\vspace{-4mm}
\subsection{Including guidance in the fixpoint algorithm.}

In Fig.~\ref{fig:alg_prec} we present a version of \PLAIs (from
Fig.~\ref{alg:naive}) that includes our proposed modifications to apply
assertions during analysis. The additions to the algorithm are calls to
functions {\tt apply\_succ} and {\tt apply\_call}, that guide analysis results
with the information of the assertion conditions, and $\errors$, an
analysis-like set of triples representing inferred states before applying the
assertions that will be used to check whether the assertions provided by the
user could be proved by the analyzer (see Sect.~\ref{sec:errors}).
Success conditions are applied ({\tt apply\_succ}) after the body of the clause
has been abstractly executed. It
receives an atom $G$ and $\lambda^c$ as parameters to decide correctly
which success conditions have to be applied.
Call conditions are applied ({\tt apply\_call}) before calling function {\tt
  solve}. Otherwise, a less precise call pattern will be captured during the
procedure (it adds new entries to the table).
The last addition, $\errors$,
collects tuples to be used later to
check that the assertions were correct (see Sect.~\ref{sec:errors}).
We collect all success constraints before applying any success
conditions (line \ref{alg:esucc} of Fig.~\ref{fig:alg_prec}) and
all call constraints before applying any call condition (line
\ref{alg:ecalls} of Fig.~\ref{fig:alg_prec}).

\begin{figure}[t]
  \textsc{Algorithm} \textbf{\textsc{Guided\_analyze}}($P, \Q_\alpha$)
  
  \textbf{input:} $P, \Q_\alpha$
  \textbf{global output:} $A \leftarrow \emptyset, \trustf{\errors} \leftarrow \emptyset$
  \begin{algorithmic}[1]
    \State $a[G_i,\lambda_i] \leftarrow \bot$ {\bf for all}
    $\CDP{G_i}{\lambda_i} \in \{\CDP{G}{\lambda^t} | \lambda^t =$ \trustf{\tt
      apply\_call}$(G,\lambda), \CDP{G}{\lambda} \in \Q_\alpha\}$

    \State {$\changes \leftarrow$ true}

    \While{$\changes$}
    \State $\changes \leftarrow$ false
    \State $W \leftarrow \{(G,\lambda^c,\cl)\ |\
    a[G,\lambda^c]$ is defined $\wedge\ \cl \in P\ \wedge \exists
    \sigma$ s.t. $G = \sigma(\cl.{\tt head})\}$

    \ForEach{$(G,\lambda^c,\cl) \in W$}
    \State $\lambda^t \leftarrow$ {\tt abs\_call}$(G,\lambda^c,\cl.{\tt head})$
    \State $\lambda^t \leftarrow$ {\tt
      solve\_body}$(\cl.{\tt body}, \lambda^t)$
    \State $\lambda^{s_0} \leftarrow {\tt abs\_proceed}(G,\cl.{\tt
      head},\lambda^t)$
    \State $\lambda^{s_1} \leftarrow {\tt
      abs\_generalize}(\lambda^{s_0}, \{a[G,\lambda^c]\})$
    \State \label{alg:esucc}
    \trustf{$\errors \leftarrow \errors \cup \{\tuple{G,\lambda^c,\lambda^{s_1}}\}$}
    \State {\trustf{$\lambda^{s}$} $\leftarrow$ \trustf{{\tt apply\_succ}}
      $(G,\lambda^c,\lambda^{s_1})$} 
    \If{\trustf{$\lambda^{s}$} $\neq a[G,\lambda^c]$}
    \State   $a[G,\lambda^c] \leftarrow {\color{orange} \lambda^{s}}$, \label{alg:ins_succ}
        $\changes \leftarrow$ true
            \Comment{Fixpoint not reached yet}
    \EndIf
    \EndFor 
    \EndWhile
    
    \Function{\tt solve\_body}{$B, \lambda^t$}
    \ForEach {$L \in B$}
    \State {$\lambda^c \leftarrow$ {\tt \project}$(L,\lambda^t)$}
    \State $\call = \{\lambda\ |\ a[H,\lambda'] $ is defined $\wedge\ \exists
    \sigma$ s.t. $\sigma(H)=L \wedge \lambda = \sigma(\lambda') \}$
    \State $\lambda^{c'} \leftarrow$ {\tt abs\_generalize}$(\lambda^c, \call)$
    \State \label{alg:ecalls} \trustf{$\errors \leftarrow \errors \cup \{\tuple{L,\lambda^{c'},\_}\}$}
    \State {\trustf{$\lambda^{c'}$} $\leftarrow$ \trustf{{\tt
          apply\_call}}$(L,\lambda^c)$}
    \State $\lambda^s \leftarrow $ {\tt solve}$(L,{\color{orange}\lambda^{c'}})$
    \State $\lambda^t \leftarrow$ {\tt \extend}$(L,\lambda^s,\lambda^t)$
    \EndFor
    \State \Return $\lambda^t$
    \EndFunction
  \end{algorithmic}
\vspace*{-2mm}
  \caption{Fixpoint analysis algorithm using \trustf{\bf assertion conditions}.}
  \label{fig:alg_prec}
\end{figure}

\medskip
Assuming that we are analyzing program $P$ and the applicable assertion
conditions are stored in \conds, the correct application of assertions is
described in Fig.~\ref{fig:apply_assrt}. Flag {\tt \accel} controls if
assertions are used to recover accuracy or to (possibly) speed up fixpoint
computation.

\noindent
\textbf{Applying call conditions.}
Given an atom $G$ and an abstract
call constraint $\lambda^c$, if there is a call assertion condition for
$G$, if {\tt \accel} is true, $\tsp{\pre, P}$ is used
directly, otherwise the operation $\tsp{\pre, P} \sqcap \lambda^c$ will
prune from the analysis result the (abstracted) states that are outside the
precondition. An over-approximation has to be made, otherwise we may remove
calling states that the user did not specify.

\begin{figure}[t]
  \textbf{global flag:} {\tt \accel}
  \begin{algorithmic}[1]
    \Function {\tt apply\_call}{$L, \lambda^c$}
    \If {$\exists \sigma, \lambda^t =\tsp{\sigma(\pre), P}$ s.t.
      $\calls(H,\pre) \in \conds, \sigma(H) =L$} \label{alg:call_cond}
    \State {\bf if} {{\tt \accel}} {\bf return} $\lambda^t$ {\bf else
      return}
    $\lambda^c \sqcap \lambda^t$ \label{alg:apply_call}

    \Else \ \Return $\lambda^c$
    \EndIf
    
    \EndFunction

    \Function {\tt apply\_succ}{$G, \lambda^c, \lambda^{s_0}$}
    \State $app= \{\lambda\ |\ \exists\
    \sigma, \success(H,\pre,\post) \in \conds, \sigma(H) = G,$
    \\ \hspace{20mm}
    $\lambda = \tsp{\sigma(\post), P}, \tsl{\sigma(\pre), P} \sqsupseteq \lambda^c\}$
        
    \If {$app \neq \emptyset$}
    \State $\lambda^t = \bigsqcap app$
    \State {\bf if} {\texttt{\accel}} {\bf return} $\lambda^t$ {\bf else return}
    $\lambda^t \sqcap \lambda^{s_0}$
    \Else\ \Return $\lambda^{s_0}$
    \EndIf
    \EndFunction
  \end{algorithmic}
\vspace*{-2mm}
  \caption{Applying assertions.}
  \label{fig:apply_assrt}
\vspace*{-5mm}
\end{figure}

\noindent
\textbf{Applying success conditions.} Given an atom $G$, an abstract call
constraint $\lambda^c$ and its corresponding abstract success constraint
$\lambda^s$, all success conditions whose precondition applies ($\lambda^c
\sqsubseteq \lambda^-_{TS(\pre, P)}$) are collected in $app$. Making an
under-approximation of \pre is necessary to consider the application of the
assertion condition only if it would be applied in the concrete executions of
the program. An over-approximation of \post needs to be performed since
otherwise success states that actually happen in the concrete execution of the
program may be removed. If no conditions are applicable (i.e., $app$ is empty),
the result is kept as it was. Otherwise, if the flag {\tt \accel} is true
$\tsp{\post,P}$ is used, as it is; otherwise, it is used to refine the value of
the computed answer $\lambda^s$.

\noindent
Applying assertion conditions bounds the extrapolation (widening) performed by
{\tt abs\_generalize}, avoiding unnecessary precision losses. %
Note that the existence of guidance assertions for a predicate does not save
having to analyze the code of the corresponding predicate if it is available,
since otherwise any calls generated within that predicate would be omitted and
not analyzed for, resulting in an incorrect analysis result.

\vspace{-4mm}
\subsection{Fundamental properties of analysis guided by assertions}

We claim the following properties for analysis of a program $P$ applying
assertions as described in the previous sections. The inferred abstract
execution states are covered by the call and (applicable) success assertion
conditions.

\begin{lemma}{\textbf{Applied call conditions.}}\label{th:call}
  Let $\calls(H,\pre)$ be an assertion condition from program $P$, and let
  $\tuple{G,\lambda^c, \lambda^s}$ be a triple derived for $P$ and initial
  queries $\Q_\alpha$ by \textsc{Guided\_analyze}$(P,Q_\alpha)$. If
  $G=\sigma(H)$ for some renaming $\sigma$ then $\lambda^c \sqsubseteq
  \tsp{\sigma(\pre), P}$.  
\end{lemma}

\begin{proof}
  Function {\tt apply\_call} obtains in $\lambda^t$ the trusted value for the
  call. It restricts the encountered call $\lambda^c$ or uses it as is, in any
  case $\lambda^c \sqsubseteq \lambda^t = \tsp{\pre, P}$. Hence if this function
  is applied whenever inferred call patterns are introduced in the analysis
  results, the lemma will hold.
  
  The lemma holds after initialization, since the function is applied before
  inserting the tuples in $A$. Now we reason about how the algorithm changes the
  results. The two spots in which analysis results are updated are in function
  {\tt solve} (line~\ref{alg:ins_call} of Fig.~\ref{fig:solve}) and in the body
  of the loop of the algorithm (line~\ref{alg:ins_succ} of
  Fig.~\ref{fig:alg_prec}).
Function {\tt solve} adds tuples to the analysis whenever new encountered call
patterns are found, it is called right after {\tt apply\_call},
therefore it only inserts call patterns taking into account calls conditions.
The analysis updates made in the body of the loop do not
insert new call patterns, only the recomputed success abstractions for those
already present (previously collected in $W$), therefore all call patterns
encountered are added taking into account the call conditions and the lemma
holds.
$\qed$
\end{proof}

\begin{lemma}{\textbf{Applied success conditions.}}\label{th:succ}
  Let $\success(H,\pre,\post)$ be an assertion condition from program $P$ and
  let $\tuple{G,\lambda^c,\lambda^s}$ be a triple derived for $P$ with
  $\Q_\alpha$ initial queries by \textsc{Guided\_analyze}$(P,\Q_\alpha)$ .
  If $G = \sigma(H)$ for some renaming $\sigma$ then $\lambda^c \sqsubseteq$
  $\tsl{\sigma(\pre), P} \Rightarrow \lambda^s \sqsubseteq \tsp{\sigma(\post),
    P}$.
\end{lemma}
\begin{proof}
  Function {\tt apply\_succ} computes the $\sqcap$ of all applicable assertion
  conditions (checking $\lambda^c \sqsubseteq \tsl{\pre,P}$), if existing. Since
  we make the $\sqcap$ of all applied conditions, $\lambda^s \sqsubseteq
  \bigsqcap \tsp{\post_i, P} \sqsubseteq \tsp{\post, P}$ for any \post. Hence if
  all results inserted in the analysis result have been previously processed by
  {\tt apply\_succ} the lemma holds.
  The lemma holds for the initialized results, because $\lambda^s = \bot
  \sqsubseteq \tsp{\post, P}$ for any \post.
  Now we reason about how the algorithm changes the results. We have the same
  points in the algorithm that change the analysis result as in the proof of
  Lemma~\ref{th:call}. The {\tt solve} function initializes $\lambda^s$ of the
  newly encountered calls with $\bot$, so it is the same situation as when
  initializing. In the body of the loop {\tt apply\_succ} is always called
  before updating the value in the result and the lemma holds. $\qed$
\end{proof}

\vspace{-4mm}
\section{Checking correctness in a guided analysis}
\label{sec:errors}
\vspace{-2mm}

We discuss how assertions may introduce errors in the analysis, depending on
their status.
\sample assertions are not used by the analyzer. Any part of the execution
stopped by them will conservatively be considered to continue, keeping the
analysis safe.
{\tt check} assertions stop the execution of the program if
the properties of the conditions are not met. Hence it is safe to
narrow the analysis results using their information. 
Last, {\tt trust} assertions are not considered during the concrete executions,
so they may introduce errors. Such assertion conditions express correct
properties if they comply with the following definitions:

\vspace*{-2mm}
\begin{definition}[Correct call condition]\label{def:corr_call}
  Let P be a program with an assertion condition $C = \calls(H,\pre)$. $C$ is
  correct for a query $\Q$ to $P$ if for any predicate descriptor $G$, s.t. $G =
  \sigma(H)$ for some renaming $\sigma$, $\forall \theta^c \in \callingcontext,$ \\
  $\theta^c \in \gamma(\tsp{\sigma(\pre),P})$.
\end{definition}

\begin{definition}[Correct success condition]\label{def:corr_succ}
  Let $P$ be a program with an assertion condition $C = \success(H,\pre,\post)$.
  $C$ is correct for $P$ if for any predicate descriptor $G$, s.t. $G =
  \sigma(H)$ for some renaming $\sigma$, $\theta^c \in \gamma(\tsl{\sigma(\pre), P}),
  \theta^s \in \answers(P,\{\CDP{G}{\theta^c}\}) \Rightarrow \theta^s \in  \gamma(\tsp{\sigma(\post), P})$.
\end{definition}

\begin{theorem}{{\bf Correctness modulo assertions.}}\label{th:full}
  Let $P$ be a program with \emph{correct} assertion conditions $C$ and
  $\Q_\alpha$ a set of initial abstract queries. Let $\Q$ be the set of concrete
  queries: $\Q = \{\CDP{G}{\theta} \mid\ \theta\in\gamma(\lambda) \wedge
  \CDP{G}{\lambda} \in \Q_\alpha\}$.

  \noindent
  The analysis result $A = \{\tuple{G_1,\lambda_1^c, \lambda_1^s}, \ldots,
  \tuple{G_n,\lambda_n^c,\lambda_n^s}\}$ computed with \\
  \textsc{Guided\_analyze}$(P, \Q_\alpha)$
  is \emph{correct} (Def.~\ref{def:correct}) for $P, \Q$.
\end{theorem}
\newcommand{\CCP}{\ensuremath{\theta^c}}
\newcommand{\CP}{\ensuremath{\lambda^c}}
\newcommand{\CAP}{\ensuremath{\theta^s}}
\newcommand{\AP}{\ensuremath{\lambda^s}}
\newcommand{\atuple}{\ensuremath{\tuple{G,\CP,\AP}}}
\newcommand{\etuple}{\ensuremath{\tuple{G,\CP_E,\AP_E}}}
\newcommand{\tsla}[1]{\ensuremath{\lambda^-_{#1}}}
\newcommand{\tspa}[1]{\ensuremath{\lambda^+_{#1}}}

\begin{proof}
  For conciseness in the proof we omit the renaming part. Fixed program $P$,
  given an abstract description $d$ from an assertion (\pre or \post), let
  $\lambda^-_{d} = \tsl{d,P}, \lambda^+_{d} = \tsp{d,P}$. If there are no
  assertion conditions, the theorem trivially holds (Thm.~\ref{th:basic}). If
  assertion conditions are used to generalize, the theorem also holds because
  $\CP = \tspa{\pre}$ and $\AP = \tspa{\post}$ are by definition
  (Def.~\ref{def:corr_call}, Def.~\ref{def:corr_succ}, respectively) correct
  over-approximations. If assertion conditions are used to regain precision:

  \noindent  \textbf{Call}: We want to prove that \\
  $\forall G, \theta^c \in \callingcontext$
    $\exists \tuple{G,\lambda^c,\lambda^s}\in A$
    s.t.\ $\theta^c \in \gamma(\lambda^c)
  \text{(Def.~\ref{def:correct})}.$
  \begin{align*}
      \text{We applied: }&  \calls(G,\pre)
    \\  & \CCP \in \gamma(\tspa{\pre}) \tag{by Def.~\ref{def:corr_call}}
    \\ \text{In\ } E:\ & \exists \etuple \in E,  \CCP \in \gamma(\CP_E) \tag{by algorithm (Fig.~\ref{fig:alg_prec} line~\ref{alg:ecalls})}
    \\ \text{Then:\ } & \CCP \in \gamma(\CP_E) \cap \gamma(\tspa{\pre}) \subseteq \gamma(\alpha(\gamma(\CP_E) \cap \gamma(\tspa{\pre}))) \subseteq \gamma(\CP_E \sqcap \tspa{\pre})
    \\ & \CCP \in \gamma(\CP_E \sqcap \tspa{\pre}) = \gamma(\CP) \tag{by algorithm (Fig.~\ref{fig:apply_assrt} line~\ref{alg:apply_call})}
\end{align*}
\noindent
\textbf{Success}: We want to prove that \\
$\forall \tuple{G,\lambda^c,\lambda^s} \in A, \forall
    \theta^c\in\gamma(\lambda^c)$ if $\theta^s \in \answers(P,
    \{\CDP{G}{\theta^c}\})$
    then $\theta^s\in\gamma(\lambda^s)$. 
\begin{align*}
  \text{We applied: }&  \success(G,\pre_i, \post_i)
  \\ & \CP \sqsubseteq \tsl{\pre_i} \implies \AP \sqsubseteq \tspa{\post_i} \tag{by Lemma~\ref{th:succ}}
  \\ & \CCP \in \gamma(\tsla{\pre_i}), \CAP \in \answers(P,\{\CDP{G}{\CCP}\}) \implies \CAP \in \tspa{\post_i}  \tag{by Def.~\ref{def:corr_succ}}
  \\ & \lambda^p = \bigsqcap \{\tspa{\post}\ |\ \success(G,\pre,\post), \forall \CCP \in \lambda^c, \CCP \in \gamma(\tsla{\pre})\}
  \\ & \CCP \in \gamma(\CP),  \CAP \in \answers(P,\{\CDP{G}{\CCP}\}) \implies \CAP \in \lambda^p
  \\ & \exists \etuple \in E, \text{ s.t. }  \CP \sqsupseteq \gamma(\CP_E) \tag{unrefined abstractions}
  \\ \text{We have: } & \CAP \in \gamma(\AP_E), \CAP \in \gamma(\lambda^p)
  \\ & \CAP \in \gamma(\AP_E) \cap \gamma(\lambda^p) \subseteq  \gamma(\alpha(\gamma(\AP_E) \cap \gamma(\lambda^p))) \subseteq \gamma(\AP_E \sqcap \lambda^p)
  \\ & \CAP \in \gamma(\AP_E \sqcap \lambda^p) = \gamma(\AP)\ &\square
\end{align*}
\end{proof}

\noindent
In other words, Theorem~\ref{th:full} and Lemmas \ref{th:call} and \ref{th:succ}
ensure that correct assertion conditions bound imprecision in the result,
without affecting correctness. By applying the assertion conditions no actual
concrete states are removed from the abstractions.

We can identify suspicious pruning during analysis. Let $\lambda^a$ be the
correct approximation of a condition and $\lambda$ be an inferred abstract
state, typically a value in the tuples of $\errors$.
If $\lambda \sqcap \lambda^a = \bot$ the inferred information is incompatible
with that in the condition, therefore it is likely that the assertion is
\emph{erroneous}.
$\lambda \not\sqsubseteq \lambda^a$ indicates that the algorithm inferred more
concrete constraint states than described in the assertion and the analysis
results may be wrong.
These checks can be performed while the algorithm is run or off-line, by
comparing the properties of the assertion conditions against the triples stored
in $\errors$, which, as mentioned earlier, stores partial analysis results with
no assertions applied. A full description of this checking procedure is
described
in~\cite{assrt-theoret-framework-lopstr99-short,optchk-journal-scp-short}.

\vspace{-4mm}
\section{Related work}
\label{sec:related-work}
\vspace{-2mm}
The inference of arbitrary semantic properties of programs is known to
be both undecidable and expensive, requiring user
interaction in many realistic
settings.
Abstract interpreters allow the \textbf{selection of different domains and
parameters for such domains} (e.g., polyhedra, octagons, regtypes with
depth-k, etc.), as well as their widening operations (e.g., type
shortening, structural widening, etc.).
Other parameters include policies for partial evaluation and other
transformations (loop unrolling, inlining, slicing, etc.).
These parameters are orthogonal or complementary to the issues
discussed in this paper.
To the extent of our knowledge the use of \textbf{program-level
  annotations} 
(such as assertions) to guide abstract interpretation has not been
widely studied in the literature, contrary to their (necessary) use in
verification and theorem proving approaches.
The \textbf{Cibai}~\cite{Logozzo-vmcai07-short} system includes 
\emph{trust-style} annotations while sources are processed to encode
some predefined runtime semantics.
In~\cite{gopan2007guided-short} analysis is guided by modifying the analyzed
program to restrict some of its behaviors. However, this guidance affects the
\emph{order of program state exploration}, rather the analysis
results,
as in our case.
As mentioned in the introduction, the closest to our approach is
\textbf{Astr\'{e}e}, that allows \emph{assert}-like statements, where
correctness of the analysis is ensured by the presence of compulsory
runtime checks, and trusted (\emph{known facts}) asserts. These refine
and guide analysis operations at program points. Like in CiaoPP, the
analyzer shows errors if a known fact can be falsified
statically. However, as with the corresponding Ciao assertions, while
there has been some examples of use~\cite{DBLP:conf/sas/DelmasS07-short},
there has been no detailed description of how such assertions are
handled in the fixpoint algorithm. We argue that this paper contributes in this
direction.

\vspace{-2mm}
\section{Conclusions}
\vspace{-2mm}

We have proposed a user-guided multivariant fixpoint algorithm that makes use of
check and trust assertion information, and we have provided formal results on
the influence of such assertions on correctness and efficiency. We have extended
the semantics of the guidance (and all) assertions to deal with both the cases
in which the program execution will and will not incorporate run-time tests for
unverified assertions, as well as the cases in which the assertions are intended
for refining the information or instead to lose precision in order to gain
efficiency.
We show that these annotations are not only useful when dealing with incomplete
code but also provide the analyzer with recursion/loop invariants for speeding
up global convergence.

\begin{small}
\renewcommand\refname{\vspace{-3mm}References}
\bibliographystyle{splncs04}
\bibliography{../../../bibtex/clip/clip,../../../bibtex/clip/general}
\end{small}

\end{document}